\documentclass{sig-alternate}
\usepackage{graphicx}
\usepackage{mathrsfs}
\usepackage{balance}
\usepackage{algorithm}
\usepackage{algorithmicx}

\usepackage[noend]{algpseudocode}
\usepackage{xspace}
\usepackage{multirow}

\usepackage[nosort]{cite}

\usepackage[font=footnotesize]{subfig}
\usepackage[dvips,usenames]{color}
\graphicspath{{figures/}}

\newtheorem{theorem}{Theorem}[section]

\newtheorem{corollary}{Corollary}[section]

\newcommand{\myparagraph}[1]{\vspace{0.3em} \noindent {\textbf{#1:}}}
\newcommand{\revision}[1]{#1}

\newcommand{\abs}[1]{\left| #1 \right|}
\newcommand{\set}[1]{\left\{ #1 \right\}}
\newcommand{\ignore}[1]{}

\newcommand{\query}{\textsc{Query}}

\DeclareCaptionType{copyrightbox}
\begin{document}
\conferenceinfo{SIGMOD'13,} {June 22--27, 2013, New York, New York, USA.}
\CopyrightYear{2013}
\crdata{978-1-4503-2037-5/13/06}
\clubpenalty=10000
\widowpenalty=10000

\title{Fast Exact Shortest-Path Distance Queries on Large Networks by Pruned Landmark Labeling}

\numberofauthors{3}
\author{
% 1st. author
\alignauthor
Takuya Akiba\\
\affaddr{The University of Tokyo}\\
\affaddr{Tokyo, 113-0033, Japan}\\
\email{t.akiba@is.s.u-tokyo.ac.jp}
% 2nd. author
\alignauthor
Yoichi Iwata\\
\affaddr{The University of Tokyo}\\
\affaddr{Tokyo, 113-0033, Japan}\\
\email{\hspace{-0.05em}y.iwata@is.s.u-tokyo.ac.jp\hspace{-0.05em}}
% 3rd. author
\alignauthor
Yuichi Yoshida\\
\affaddr{\hspace{-0.4em}National Institute of Informatics,\hspace{-0.4em}}\\
\affaddr{\hspace{-0.1em}Preferred Infrastructure, Inc.\hspace{-0.1em}}\\
\affaddr{Tokyo, 101-8430, Japan}\\
\email{yyoshida@nii.ac.jp}
}

\maketitle

\begin{abstract}
We propose a new exact method for shortest-path distance queries on large-scale networks.
Our method precomputes distance labels for vertices by performing a breadth-first search from every vertex.
Seemingly too obvious and too inefficient at first glance,
the key ingredient introduced here is \textit{pruning} during breadth-first searches.
While we can still answer the correct distance for any pair of vertices from the labels,
it surprisingly reduces the search space and sizes of labels.
Moreover, we show that we can perform 32 or 64 breadth-first searches simultaneously
exploiting bitwise operations.
We experimentally demonstrate that
the combination of these two techniques is efficient and robust on various kinds of large-scale real-world networks.
In particular, our method can handle social networks and web graphs with hundreds of millions of edges,
which are two orders of magnitude larger than the limits of previous exact methods,
with comparable query time to those of previous methods.

\end{abstract}


\section{Introduction}
\label{sec:introduction}
A \emph{distance query} asks the distance between two vertices in a graph.
Without doubt,
answering distance queries is one of the most fundamental operations on graphs,
and it has wide range of applications.
For example, on social networks,
distance between two users is considered to indicate the closeness,
and used in socially-sensitive search to help users to find more related users or
 contents~\cite{practice/landmarks/Vieira2007, apps/social_search/vldb2008},
or to analyze influential people and communities~\cite{apps/social_analysis/kdd03, apps/social_analysis/kdd06}.
On web graphs, distance between web pages is one of indicators of relevance,
and used in context-aware search to give higher ranks
to web pages more related to the currently visiting web
page~\cite{apps/context_search/cikm2008, practice/landmarks/Potamias2009}.
Other applications of distance queries include
top-$k$ keyword queries on linked data~\cite{apps/topk/sigmod07, apps/topk/icde09},
discovery of optimal pathways between compounds in metabolic
networks~\cite{apps/bio/Rahman2005,apps/bio/Rahman2006},
and management of resources in computer
networks~\cite{apps/comm/Pastor-Satorras2004,apps/comm/Boccaletti2006}.

Of course, we can compute the distance for each query by using a breadth first search (BFS) or Dijkstra's algorithm.
However, they take more than a second for large graphs,
which is too slow to use as a building block of these applications.
\revision{In particular,
applications such as socially-sensitive search or context-aware search
should have low latency since they involve real-time interactions between users,
while they need distances between a number of pairs of vertices to rank items for each search query.
Therefore, distance queries should be answered much more quickly, say, microseconds.}

The other extreme approach is to compute distances between all pairs of vertices beforehand and store them in an index.
Though we can answer distance queries instantly,
this approach is also unacceptable since preprocessing time and index size are quadratic and unrealistically large.
Due to the emergence of huge graph data,
design of more moderate and practical methods between these two extreme approaches
has been attracting strong interest in the database community~\cite{practice/exact/Cheng2009, practice/landmarks/Potamias2009, practice/exact/Wei2010, practice/landmarks/Tretyakov2011, practice/exact/Akiba2012, practice/landmarks/Qiao2012, practice/exact/Jin2012}.

Generally, there are two major graph classes of real-world networks:
one is road networks,
and the other is complex networks such as social networks, web graphs, biological networks and computer networks.
For road networks,
since it is easier to grasp and exploit structures of them,
research has been already very successful.
Now distance queries on road networks can be processed in less than one microsecond
for the complete road network of the USA~\cite{practice/road/Abraham2011}.

In contrast, answering distance queries on complex networks is still a highly challenging problem.
The methods for road networks do not perform well on these networks
since structures of them are totally different.
Several methods have been proposed for these networks,
but they suffer from drawback of scalability.
They take at least thousands of seconds or tens of thousands of seconds
to index networks with millions of edges~\cite{practice/exact/Wei2010, practice/exact/Akiba2012, practice/exact/Abraham2012, practice/exact/Jin2012}.

To handle larger complex networks,
apart from these exact methods,
approximate methods are also studied.
That is, we do not always have to answer correct distances.
They are successful
in terms of much better scalability and very small average relative error for random queries.
However,
some of these methods take milliseconds to answer queries~\cite{practice/landmarks/Gubichev2010, practice/landmarks/Tretyakov2011, practice/landmarks/Qiao2012},
which is about three orders of magnitude slower than other methods.
Some other methods answer queries in microseconds~\cite{practice/landmarks/Potamias2009, practice/landmarks/Vieira2007},
but it is reported that
precision of these methods for close pairs of vertices is not high~\cite{practice/landmarks/Qiao2012, practice/exact/Akiba2012}.
This drawback might be critical for applications such as socially-sensitive search or context-aware search since,
in these applications,
distance queries are employed to distinguish close items.

\subsection{Our Contributions}
To address these issues, in this paper,
we present a new method for answering distance queries in complex networks.
The proposed method is an exact method.
That is, it always answers exactly correct distance to queries.
It has much better scalability than previous exact methods and can handle graphs with hundreds of millions of edges.
Nevertheless, the query time is very small and around ten microseconds.
Though our method can handle directed and/or weighted graphs as we mention later,
in the following, we assume undirected, unweighted graphs for simplicity of exposition.

Our method is based on the notion of \textit{distance labeling} or \textit{distance-aware 2-hop cover}.
The idea of 2-hop cover is as follows.
For each vertex $u$,
we pick up a set $C(u)$ of candidate vertices so that every pair of vertices $(u,v)$ has at least one vertex $w \in C(u) \cap C(v)$ on a shortest path between the pair.
For each vertex $u$ and a vertex $w \in C(u)$,
we precompute the distance $d_G(u,w)$ between them.
We say that the set $L(u) = \{(w,d_G(u,w))\}_{w \in C(u)}$ is the \textit{label} of $u$.
Using labels,
it is clear that the distance $d_G(u,v)$ between two vertices $u$ and $v$ can be computed as $\min\{\delta + \delta' \mid (w,\delta) \in L(u), (w,\delta') \in L(v)\}$.
The family of labels $\{L(u)\}$ is called a \textit{2-hop cover}.
Distance labeling is also commonly used in
previous exact methods~\cite{practice/exact/Cohen2002, practice/exact/Cheng2009, practice/exact/Abraham2012, practice/exact/Jin2012},
but we propose a totally new and different approach
to compute the labels, referred to as the \textit{pruned landmark labeling}.

The idea of our method is simple and rather radical:
from every vertex,
we conduct a breadth-first search and add the distance information to labels of visited vertices.
Of course, if we naively implement this idea,
we need $O(nm)$ preprocessing time and $O(n^2)$ space to store the index, which is unacceptable.
Here, $n$ is the number of vertices and $m$ is the number of edges.
Our key idea to make this method practical is \textit{pruning} during the breadth-first searches.
Let $S$ be a set of vertices
and suppose that we already have labels that can answer
correct distance between two vertices if a shortest path between them passes through a vertex in $S$.
Suppose we are conducting a BFS from $v$ and visiting $u$.
If there is a vertex $w \in S$ such that $d_G(v, u) = d_G(v, w) + d_G(w, u)$,
then we \textit{prune} $u$.
That is, we do not traverse any edges from $u$.
As we prove in Section~\ref{sec:algorithm:proof},
after this pruned BFS from $v$,
the labels can answer the distance between two vertices if a shortest path between them
passes through a vertex in $S \cup \set{v}$.

Interestingly, our method combines the advantages of
three different previous successful approaches:
landmark-based approximate methods~\cite{practice/landmarks/Potamias2009, practice/landmarks/Tretyakov2011, practice/landmarks/Qiao2012},
tree-decomposition-based exact methods~\cite{practice/exact/Wei2010, practice/exact/Akiba2012},
and labeling-based exact methods~\cite{practice/exact/Cohen2002, practice/exact/Cheng2009, practice/exact/Abraham2012}.
Landmark-based approximate methods achieve remarkable precision
by leveraging the existence of highly \textit{central} vertices in complex networks~\cite{practice/landmarks/Potamias2009}.
This fact is also the main reason of the power of our pruning:
by conducting breadth-first searches from these central vertices first,
later we can drastically prune breadth-first searches.
Tree-decomposition-based methods
exploit the core--fringe structure
of networks~\cite{rel/core_fringe/physical00, rel/core_fringe/physical01}
by decomposing tree-like fringes of low tree-width.
Though our method does not explicitly use tree decompositions,
we prove that our method can efficiently process graphs of small tree-width.
This process indicates that our method also exploits the core--fringe structure.
As with other labeling-based methods,
the data structure of our index is simple
and query processing is very quick because of the locality of memory access.

{
\tabcolsep = 1.3mm
\begin{table}[tbp]
\center
\small
\caption{Summary of experimental results of previous methods and the proposed method for exact distance queries.
}

\begin{tabular}{|c|crr|rr|}
\hline
Method & Network & $|V|$ & $|E|$ & Indexing & Query \\
\hline\hline
TEDI & Computer & 22 K & 46 K & 17 s & 4.2 $\mu$s \\
\cite{practice/exact/Wei2010} & Social & 0.6 M & 0.6 M & 2,226 s & 55.0 $\mu$s \\
\hline
HCL & Social & 7.1 K & 0.1 M & 1,003 s & 28.2 $\mu$s \\
\cite{practice/exact/Jin2012} & Citation & 0.7 M & 0.3 M & 253,104 s & 0.2 $\mu$s \\
\hline
TD & Social & 0.3 M & 0.4 M & 9 s & 0.5 $\mu$s \\
\cite{practice/exact/Akiba2012} & Social & 2.4 M & 4.7 M & 2,473 s & 0.8 $\mu$s \\
\hline
HHL & Computer & 0.2 M & 1.2 M & 7,399 s & 3.1 $\mu$s \\
\cite{practice/exact/Abraham2012} & Social & 0.3 M & 1.9 M & 19,488 s & 6.9 $\mu$s \\
\hline
 & Web & 0.3 M & 1.5 M & \textbf{4 s} & 0.5 $\mu$s \\
\textbf{PLL} & Social & 2.4 M & 4.7 M & \textbf{61 s} & 0.6 $\mu$s \\
\cline{2-6}
\textbf{(this work)} & Social & 1.1 M & \textbf{114 M} & 15,164 s & 15.6 $\mu$s \\
 & Web & 7.4 M & \textbf{194 M} & 6,068 s & 4.1 $\mu$s \\

\hline

\end{tabular}
\label{tbl:methods}
\end{table}
}

Though this pruned landmark labeling scheme is already powerful by itself,
we propose another labeling scheme with a different kind of strength
and combine them to further improve the performance.
That is, we show that labeling by breadth-first search can be implemented in a bit-parallel way,
which exploits the property that the number of bits $b$ in a register word is typically 32 or 64
and we can perform bit manipulations on these $b$ bits simultaneously.
By this technique,
we can perform BFSs from $b + 1$ vertices simultaneously in $O(m)$ time.
In the beginning,
this bit-parallel labeling (without pruning) works better than the pruned landmark labeling since pruning does not happen much.
Note that we are not talking about thread-level parallelism,
and our bit-parallelism actually decreases the computational complexity by the factor of $b+1$.
We can also use thread-level parallelism in addition to these two labeling schemes.

As we confirm in our experimental results,
our method outperforms other state-of-the-art methods for exact distance queries.
In particular, it has significantly better scalability than previous methods.
It took only tens of seconds for indexing networks with millions of edges.
This indexing time is two orders of magnitude faster than previous methods,
which took at least thousands of seconds or even more than one day.
Moreover, our method successfully handled networks with hundreds of millions of edges,
which is again two orders of magnitude larger than networks that
have been previously used in experiments of exact methods.
The query time is also better than previous methods for networks with the same size,
and we confirmed that the query time does not increase rapidly against sizes of networks.
We also confirm the size of an index of our method is comparable to other methods.

In Table~\ref{tbl:methods}, we summarize our experimental results
and those of previous exact methods presented in these papers.
We listed the results for the largest two real-world complex networks from each paper.
In our experiments,
we further compare our method
with hierarchical hub labeling~\cite{practice/exact/Abraham2012} and
the tree-decomposition-based method~\cite{practice/exact/Akiba2012}.

In Section~\ref{sec:related_work},
we describe related works on exact and approximate distance queries.
In Section~\ref{sec:preliminaries},
we give definitions and notions used in this paper.
Section~\ref{sec:algorithm} is devoted to explain our first scheme,
the pruned landmark labeling.
We explain our second scheme, the bit-parallel labeling,
in Section~\ref{sec:bit-parallel}.
In Section~\ref{sec:variants},
we mention variants of distance queries we can handle by slightly modifying our method.
We explain our experimental results in Section~\ref{sec:experiments},
and conclude in Section~\ref{sec:conclusions}.

\section{Related Work}
\label{sec:related_work}
\subsection{Exact Methods}
\label{sec:related_work:exact}

For exact distance queries on complex networks such as social networks and web graphs,
several methods are proposed recently.

Large portion of these methods can be considered as
based on the idea of 2-hop cover~\cite{practice/exact/Cohen2002}.
Finding small 2-hop covers efficiently is a challenging and long-standing problem~\cite{practice/exact/Cohen2002, practice/exact/Cheng2009, practice/exact/Abraham2012}.
One of the latest methods is \textit{hierarchical hub labeling}~\cite{practice/exact/Abraham2012},
which is based on a method for road networks~\cite{practice/road/Abraham2011}.
Another latest method related to 2-hop cover is \textit{highway-centric labeling}~\cite{practice/exact/Jin2012}.
In this method,
we first compute a spanning tree $T$ and use it as a ``highway''.
That is,
when computing distance $d_G(u,v)$ between two vertices $u$ and $v$,
we output the minimum over $d_G(u,w_1) + d_T(w_1,w_2) + d_G(w_2,v)$ where $w_1$ and $w_2$ are vertices in labels of $u$ and $v$, respectively,
and $d_T(\cdot,\cdot)$ is the distance metric on the spanning tree $T$.

An approach based on \textit{tree decompositions} is also reported to be efficient~\cite{practice/exact/Wei2010, practice/exact/Akiba2012}.
A tree decomposition of a graph $G$ is a tree $T$ with each vertex associated with a set of vertices in $G$, called a \emph{bag}~\cite{theory/td/Robertson1984}.
Also, the set of bags containing a vertex in $G$ forms a connected component in $T$.
It heuristically computes a tree decompositions and stores shortest-distance matrices for each bag.
It is not hard to compute distances from this information.
The smaller the size of the largest bag is, the more efficient this method is.
Because of the core--fringe structure of the networks~\cite{rel/core_fringe/physical00
, rel/core_fringe/physical01
},
these networks can be decomposed into one big bag and many small bags,
and the size of the largest bag is moderate though not small.

\subsection{Approximate Methods}
\label{sec:rel_approx}
To gain more scalability than these exact methods, approximate methods,
which do not always answer correct distances, also have been studied.

The major approach is the
\textit{landmark-based approach}~\cite{practice/landmarks/Tang2003, practice/landmarks/Vieira2007}.
The basic idea of these methods is to select a subset $L$ of vertices as landmarks,
and precompute the distance $d_G(\ell,u)$ between each landmark $\ell \in L$ and all the vertices $u \in V$.
When the distance between two vertices $u$ and $v$ is queried,
we answer the minimum $d_G(u,\ell) + d_G(\ell, v)$ over landmarks $\ell \in L$ as an estimate.
Generally, the precision for each query depends on whether actual shortest paths pass nearby the landmarks.
Therefore, by selecting central vertices as landmarks,
the accuracy of estimates becomes much better than selecting landmarks randomly~\cite{practice/landmarks/Potamias2009, theory/power_law}.
However, for close pairs, the precision is still much worse than the average,
since lengths of shortest paths between them are small and they are unlikely to pass nearby the landmarks~\cite{practice/exact/Akiba2012}.

To further improve the accuracy,
several techniques were proposed~\cite{practice/landmarks/Gubichev2010, practice/landmarks/Tretyakov2011, practice/landmarks/Qiao2012}.
They typically store shortest-path trees rooted at the landmarks
instead of just storing distances from the landmarks.
To answer queries, they extract paths from the shortest-path trees as candidates of shortest-paths,
and improve them by finding loops or shortcuts.
While they significantly improve the accuracy,
the query time becomes up to three orders of magnitude slower.

\section{Preliminaries}
\label{sec:preliminaries}
\subsection{Notations}
Table~\ref{tbl:notations} lists the notations that are frequently used in this paper.
In this paper, we mainly focus on networks that are modeled as graphs.
Let $G=(V, E)$ be a graph with vertex set $V$ and edge set $E$.
We use symbols $n$ and $m$ to denote the number of vertices $|V|$ and the number of edges $|E|$, respectively,
when the graph is clear from the context.
We also denote the vertex set of $G$ by $V(G)$ and the edge set of $G$ by $E(G)$.
We denote the neighbors of a vertex $v \in V$ by $N_G(v)$.
That is, $N_G(v) = \set{u \in V \mid (u, v) \in E}$.

Let $d_G(u,v)$ denote the distance between vertices $u,v$.
If %
$u$ and $v$ are disconnected in $G$,
we define $d_G(u, v) = \infty$.
The distance in graphs is a metric, thus it satisfies the triangle inequalities.
That is, for any three vertices $s$, $t$ and $v$,
\begin{align}
\label{eq:triangle1}
d_G(s, t) &\leq d_G(s, v) + d_G(v, t), \\
\label{eq:triangle2}
d_G(s, t) &\geq \abs{d_G(s, v) - d_G(v, t)}.
\end{align}

We define $P_G(s, t) \subseteq V$
as the set of all vertices on the shortest paths between vertices $s$ and $t$.
In other words,
$$
P_G(s, t) = \set{v \in V \mid d_G(s, v) + d_G(v, t) = d_G(s, t)}.
$$

\subsection{Problem Definition}
This paper studies the following problem:
given a graph $G$, construct an index to efficiently answer distance queries,
which asks the distance between an arbitrary pair of vertices.

For simplicity of exposition, we mainly consider undirected, unweighted graphs.
However, our algorithm can be easily extended for directed and/or weighted graphs,
and we discuss about this extension in Section~\ref{sec:variants}.
Furthermore, our method can answer not only distances but also shortest-paths.
This extension is also discussed in Section~\ref{sec:variants}.

\iftrue
\begin{table}[tb]
\small
\centering
\caption{Frequently used notations.}
\label{tbl:notations}
\centering
\small
\begin{tabular}{|l|l|}
\hline
Notation & Description \\
\hline
\hline
$G = (V, E)$ & A graph \\
$n$ & Number of vertices in graph $G$ \\
$m$ & Number of edges in graph $G$ \\
$N_G(v)$ & Neighbors of vertex $v$ in graph $G$ \\
$d_G(u, v)$ & Distance between vertex $u$ and $v$ in graph $G$\\
\multirow{2}*{$P_G(u, v)$} & Set of all the vertices on the shortest paths \\
           & between vertex $u$ and $v$ in graph $G$\\
\hline
\end{tabular}
\end{table}
\fi

\subsection{Labels and 2-Hop Cover}
\label{sec:preliminaries_2hop}
The general framework of \textit{2-hop cover}~\cite{practice/exact/Cohen2002, practice/exact/Cheng2009, practice/exact/Abraham2012},
or sometimes called a \textit{labeling method}, is as follows.
Our method also follows this framework.

For each vertex $v$, we precompute a \textit{label} denoted as $L(v)$,
which is a set of pairs $(u, \delta_{uv})$,
where $u$ is a vertex and $\delta_{uv} = d_G(u, v)$.
We sometimes call the set of labels $\{L(v)\}_{v \in V}$ as an \textit{index}.
To answer a distance query between vertices $s$ and $t$,
we compute and answer $\query(s, t, L)$ defined as follows,
\begin{align*}
&\textsc{Query}(s, t, L) = \\
&\quad\quad \min\left\{ \delta_{vs} + \delta_{vt} \mid (v, \delta_{vs}) \in L(s), (v, \delta_{vt}) \in L(t) \right\}.
\end{align*}
We define $\textsc{Query}(s, t, L) = \infty$
if $L(s)$ and $L(t)$ do not share any vertex.
We call $L$ a \textit{(distance-aware) 2-hop cover} of $G$
if $\query(s, t, L) = d_G(s, t)$ for any pair of vertices $s$ and $t$.

For each vertex $v$,
we store the label $L(v)$ so that pairs in it are sorted by their vertices.
Then, we can compute $\query(s, t, L)$ in $O(\abs{L(s)} +\abs{L(t)})$ time using a merge-join-like algorithm.

\section{Algorithm Description}
\label{sec:algorithm}
\subsection{Naive Landmark Labeling}
\label{sec:algo_naive}
We start with the following naive method.
As the index, we conduct a BFS from each vertex and store distances between all pairs.
Though this method is too obvious and inefficient,
for the exposition of the next method, we explain the details.

Let $V = \{v_1, v_2, \ldots, v_n\}$.
We start with an empty index $L_0$,
where $L_0(u) = \emptyset$ for every $u \in V$.
Suppose we conduct BFSs from vertices in the order of $v_1, v_2, \ldots, v_n$.
After the $k$-th BFS from a vertex $v_k$,
we add distances from $v_k$ to labels of reached vertices, that is,
$L_{k}(u) = L_{k-1}(u) \cup \set{(v_k, d_G(v_k, u))}$ for each $u \in V$ with $d_G(v_k, u) \neq \infty$.
We do not change labels for unreached vertices,
that is, $L_k(u) = L_{k-1}(u)$ for every $u \in V$ with $d_G(v_k, u) = \infty$.

$L_n$ is the final index.
Obviously $\textsc{Query}(s, t, L_n) = d_G(s, t)$ for any pair of vertices $s$ and $t$, and therefore,
$L_n$ is a correct 2-hop cover for exact distance queries.
This is because, if $s$ and $t$ are reachable,
then $(s,  0) \in L_n(s)$ and $(s, d_G(s, t)) \in L_n(t)$ for example.

This method can be considered as a variant of landmark-based approximate methods,
which we mentioned in Section~\ref{sec:rel_approx}.
The standard landmark-based method can be regarded as
a method that precomputes $L_l$ instead of $L_n$
and estimates distance between $s$ and $t$ by $\query(s, t, L_l)$,
where $l \ll n$ is a parameter expressing the number of landmarks.

\newcommand{\FigureExample}[2]{
\subfloat[#2]{
\hspace{0.37em}
\epsfig{file=example/#1.dot.eps, height=29mm}\label{fig:example_pbfs_#1}
\hspace{0.37em}
}
\hspace{0.01em}
}
\begin{figure*}[t]
  \centering
  \FigureExample{}{First BFS from vertex 1. We visited all the vertices.}
  \FigureExample{}{Second BFS from vertex 2. We did not add labels to five vertices.}
  \FigureExample{}{Third BFS from vertex 3. We only visited the lower half of the vertices.}
  \FigureExample{}{Fourth BFS from vertex 4. This time we only visited the higher half.}
  \FigureExample{}{Fifth BFS from vertex 5. The search space was even smaller.}
\vspace{-0.5em}
  \caption{Examples of pruned BFSs. Yellow vertices denote the roots, blue vertices denote those which we visited
and labeled, red vertices denote those which we visited but pruned,
and gray vertices denote those which are already used as roots.}
  \label{fig:example_pbfs}
\end{figure*}

\subsection{Pruned Landmark Labeling}
\label{sec:algo_pruning}
Then, we introduce \textit{pruning} to the naive method.
Similarly to the method above,
we conduct \textit{pruned} BFSs from vertices in the order of $v_1, v_2, \ldots, v_n$.
We start with an empty index $L'_{0}$
and create an index $L'_{k}$ from $L'_{k-1}$ using the information obtained by the $k$-th pruned BFS from vertex $v_k$.

We prune BFSs as follows.
Suppose that we have an index $L'_{k-1}$ and we are conducting a BFS from $v_k$ to create a new index $L'_k$.
Suppose that we are visiting a vertex $u$ with distance $\delta$.
If $\textsc{Query}(v_k, u, L'_{k-1}) \leq \delta$,
then we prune $u$, that is,
we do not add $(v_k, \delta)$ to $L'_k(u)$ (i.e. $L'_k(u) = L'_{k-1}(u)$)
and we do not traverse any edge from vertex $u$.
Otherwise, we set $L'_k(u) = L'_{k-1}(u) \cup \set{(v_k, \delta)}$
and traverse all the edges from the vertex $u$ as usual.
As with the previous method,
we also set $L'_k(u) = L'_{k-1}(u)$ for all vertices $u \in V$ that were not visited in the $k$-th pruned BFS.
This algorithm, performing pruned BFSs,
is described as Algorithm~\ref{alg:pruned_bfs},
and the whole preprocessing algorithm is described as Algorithm~\ref{alg:total}.

Figure~\ref{fig:example_pbfs} shows examples of pruned BFSs.
The first pruned BFS from vertex 1 visits all the vertices (Figure~\ref{fig:example_pbfs_1}).
During the next pruned BFS from vertex 2 (Figure~\ref{fig:example_pbfs_2}), when we visit vertex 6,
since $\query(2, 6, L'_1) = d_G(2, 1) + d_G(1, 6) = 3 = d_G(2, 6)$,
we prune vertex 6 and we do not traverse edges from it. We also prune vertices 1 and 12.
As the number of performed BFSs increases,
we can confirm that the search space gets smaller and smaller
(Figure~\ref{fig:example_pbfs_3},\ref{fig:example_pbfs_4} and \ref{fig:example_pbfs_5}).

\begin{algorithm}[t]
\caption{Pruned BFS from $v_k \in V$ to create index $L'_k$.}
\label{alg:pruned_bfs}
\begin{algorithmic}[1]
\Procedure{PrunedBFS}{$G$, $v_k$, $L'_{k-1}$}
  \State $Q \gets $ a queue with only one element $v_k$.
  \State $P[v_k] \gets 0$ and $P[v] \gets \infty$ for all $v \in V(G) \setminus \set{v_k}$.
  \label{line:init}
  \State $L'_k[v] \gets L'_{k-1}[v]$ for all $v \in V(G)$.
  \While{$Q$ is not empty}
    \State Dequeue $u$ from $Q$.
    \If{\Call{Query}{$v_k$, $u$, $L'_{k-1}$} $ \leq P[u]$} \label{line:query}
      \State \textbf{continue}
    \EndIf
    \State $L'_k[u] \gets L'_{k-1}[u] \cup \set{(v_k , P[v_k])}$
    \ForAll{$w \in N_G(v)$ s.t.\ $P[w] = \infty$}
      \State $P[w] \gets P[u] + 1$.
      \State Enqueue $w$ onto $Q$.
    \EndFor
  \EndWhile
  \State \Return $L'_k$
\EndProcedure
\end{algorithmic}
\end{algorithm}

\begin{algorithm}[t]
\caption{Compute a 2-hop cover index by pruned BFS.}
\label{alg:total}
\begin{algorithmic}[1]
\Procedure{Preprocess}{$G$}
  \State $L'_0[v] \gets \emptyset$ for all $v \in V(G)$.
  \For{$k = 1, 2, \ldots, n$}
    \State $L'_k \gets$ \Call{PrunedBFS}{$G$, $v_k$, $L'_{k-1}$}
  \EndFor
  \State \Return $L'_n$
\EndProcedure
\end{algorithmic}
\end{algorithm}

\subsection{Proof of Correctness}
\label{sec:algorithm:proof}
In the following, we prove that this method computes a correct 2-hop cover index,
that is, $\textsc{Query}(s, t, L'_n) = d_G(s, t)$
for any pair of vertices $s$ and $t$.

\begin{theorem}
\label{theorem:correctness}
For any $0 \leq k \leq n$ and for any pair of vertices $s$ and $t$,
$\textsc{Query}(s, t, L'_k) = \textsc{Query}(s, t, L_k).$
\end{theorem}

\begin{proof}
We prove the theorem by mathematical induction on $k$.
Since $L'_0 = L_0$, it is true for $k=0$.
Now we assume it holds for $0, 1, \ldots, k-1$ and prove it also holds for $k$.

Let $s, t$ be a pair of vertices.
We assume these vertices are reachable in $G$,
since otherwise the answer $\infty$ can be obviously obtained.
Let $j$ be the smallest number such that
$(v_j , \delta_{v_j s}) \in L_k(s), (v_j , \delta_{v_j t}) \in L_k(t)$
and $\delta_{v_j s} + \delta_{v_j t} = \textsc{Query}(s, t, L_k)$.
We prove that $(v_j , \delta_{v_j s})$ and $(v_j , \delta_{v_j t})$
are also included in $L'_k(s)$ and $L'_k(t)$.
This immediately leads to $\textsc{Query}(s, t, L'_k) = \textsc{Query}(s, t, L_k)$.
Due to the symmetry between $s$ and $t$,
we prove $(v_j , \delta_{v_j s}) \in L'_k(s)$.

First, for any $i < j$,
we prove by contradiction that $v_i \not\in P_G(v_j, s)$.
If we assume $v_i \in P_G(v_j, s)$,
from Inequality~\ref{eq:triangle1}
\begin{align*}
\query(s, t, L_k)
&= d_G(s, v_j) + d_G(v_j, t) \\
&= d_G(s, v_i) + d_G(v_i, v_j) + d_G(v_j, t) \\
&\geq d_G(s, v_i) + d_G(v_i, t).
\end{align*}
Since $(v_i , d_G(s, v_i)) \in L_k(s)$ and $(v_i , d_G(t, v_i)) \in L_k(t)$,
this contradicts to the assumption of the minimality of $j$.
Therefore, $v_i \not\in P_G(v_j, s)$ holds for any $i < j$.

Now we prove that $(v_j, d_G(v_j, s)) \in L'_k(s)$.
Actually, we prove a more general fact: $(v_j, d_G(v_j, u)) \in L'_k(u)$ for all $u \in P_G(v_j, s)$.
Note that $s \in P_G(v_j, s)$.
Suppose that we are conducting the $j$-th pruned BFS from $v_j$ to create $L_j$.
Let $u \in P_G(v_j, s)$.
Since $P_G(v_j, u) \subseteq P_G(v_j, s)$ and $v_i \not \in P_G(v_j,s)$ for any $i < j$,
we have $v_i \not\in P_G(v_j, u)$ for any $i < j$.
Therefore, $\query(v_j, u, L'_{j-1}) > d_G(v_j, u)$ holds.
Thus, we visit all vertices $u \in P_G(v_j, s)$ without pruning,
and it follows that $(v_j, d_G(v_j, u)) \in L'_j(u) \subseteq L'_k(u)$.
\end{proof}

As a corollary, our method is proved to be an exact distance querying method
by instantiating the theorem with $k=n$.

\begin{corollary}
For any pair of vertices $s$ and $t$,
$$\query(s, t, L'_n) = d_G(s, t).$$
\end{corollary}

\subsection{Vertex Ordering Strategies}
\label{sec:vertex_order}
\subsubsection{Motivation}
In the algorithm description above,
we conducted pruned BFSs from vertices in the order of $v_1, v_2, \ldots, v_n$.
We can freely choose the order,
and moreover it turns out that the order is crucial for the performance of this method
as we will see in the experimental results presented in Section~\ref{sec:exp_ordering}.

To decide the order of vertices,
we should select \textit{central} vertices first in the sense that many shortest paths pass through these vertices.
Since we would like to prune later BFSs as much as possible,
we want to \textit{cover} larger part of pairs of vertices by earlier BFSs.
That is, the earlier labels should offer correct distances for as many pairs of vertices as possible,
and therefore the earlier vertices should be those who many shortest paths passes through.

This problem is quite similar to the problem of selecting good landmarks
for landmark-based approximate methods,
which is discussed well in~\cite{practice/landmarks/Potamias2009}.
In that problem, we also want to select good landmarks
so that many shortest path passes through these vertices or nearby vertices.

\subsubsection{Strategies}
Based on the results on landmark-based methods~\cite{practice/landmarks/Potamias2009},
we propose and examine the following three strategies.
In experiments, we basically use the \textsc{Degree} strategy,
and compare them empirically in Section~\ref{sec:exp_ordering}.

\vspace{0.3em}
\noindent\textsc{Random:}
We order vertices randomly.
We use this method as a baseline to
show the significance of other strategies.

\vspace{0.3em}
\noindent\textsc{Degree:}
We order vertices from those with higher degree.
The idea behind this strategy is that vertices with higher degree have stronger connection to many other vertices
and therefore many shortest paths would pass through them.

\vspace{0.3em}
\noindent\textsc{Closeness:}
We order vertices from those with the highest closeness centrality.
Since computing exact closeness centrality for all vertices costs $O(nm)$ time,
which is too expensive for large-scale networks,
we approximate closeness centrality by randomly sampling a small number of vertices
and computing distances from those vertices to all vertices.

\subsection{Efficient Implementation}

\subsubsection{Preprocessing (Algorithm~\ref{alg:pruned_bfs})}
\label{sec:implementation_preprocessing}

\noindent \textbf{Index:}
First, in the description above,
we treated $L'_{k-1}$ and $L'_k$ separately and explained as if we copy $L'_{k-1}$ to $L'_k$
for simplicity of explanation.
However, this copy can be easily avoided by keeping only one index
and adding labels to it after each pruned BFS.

\myparagraph{Initialization}
Another important note is to avoid $O(n)$ time initialization  for each pruned BFS.
The reason why this method is efficient is
that the search space of pruned BFSs gets much more smaller than the whole graph.
However if we spend $O(n)$ time for initialization, it would be the bottleneck.
What we want to do in the initialization is to
set all values in the array storing tentative distances as $\infty$
(Line~\ref{line:init}).
We can avoid $O(n)$ time initialization as follows.
Before we conduct the first pruned BFS,
we set all values in the array $P$ as $\infty$.
(This takes $O(n)$ time but we do this only once.)
Then, during each pruned BFS, we store all vertices we visited,
and after each pruned BFS,
we set $P[v]$ as $\infty$ for all each vertex $v$ we have visited.

\myparagraph{Arrays}
For the array storing tentative distances,
it is better to use 8-bit integers.
Since networks of our interest are small-world networks, 8-bit integers are enough to represent distances.
Using 8-bit integers, the array fits into low-level cache memories of recent computers,
resulting in the speed up by reducing cache misses.

\myparagraph{Querying}
To evaluate queries for pruning (Line~\ref{line:query}),
it is faster to use an algorithm different from the normal one since we can exploit the fact here that we issue many queries whose one endpoint is always $v_k$.
Before starting the $k$-th pruned BFS from $v_k$,
we prepare an array $T$ of length $n$ initialized with $\infty$
and set $T[w] = \delta_{w v_k}$ for all $(w, \delta_{w v_k}) \in L'_{k-1}(v_k)$.
To evaluate $\query(v_k, u, L'_{k-1})$,
for all $(w, \delta_{w u}) \in L'_{k-1}(u)$,
we compute $\delta_{w u} + T[w]$
and return the minimum.
Though normal querying algorithm takes $O(\abs{L'_{k-1}(v_k)} + \abs{L'_{k-1}(u)})$ time,
this algorithm runs in $O(\abs{L'_{k-1}(u)})$ time.
As Line~\ref{line:query} is the bottleneck of the algorithm,
this technique speeds up preprocessing by about twice.
Note that $T$ should be represented by 8-bit integers as the same reason described above,
and $O(n)$ time initialization for array $T$ should be avoided in the same way for array $P$.

\myparagraph{Prefetching}
Unfortunately, we cannot fit the index and the adjacency lists into the cache memory
for large-scale networks.
However, we can manually prefetch them to reduce the cache misses,
since vertices which we will access soon are in the queue.
Manual prefetching speeds up preprocessing by about 20\%.

\myparagraph{Thread-Level Parallelism}
As with parallel BFS algorithms~\cite{rel/parallel_bfs},
the pruned BFS algorithm can be also parallelized.
However, for simple experimental analysis
and fair comparison to previous methods,
we did not parallelize our implementation in the experiments.

\revision{
\myparagraph{Sorting Labels}
When applying merge-join-like algorithms to answer queries,
pairs in labels need to be sorted by vertices.
However, actually we do not need to sort explicitly %
by storing ranks of vertices instead of vertices.
That is, when adding a pair $(u, \delta)$ in the $i$-th pruned BFS from vertex $u$,
we add a pair $(i, \delta)$ instead.
Then, since pairs are added from vertices with lower rank to those with higher rank,
all the labels are automatically sorted.

}

\subsubsection{Querying}
\noindent \textbf{Sentinel:}
We add a dummy entry, $(n, \infty)$, to the label $L(v)$ for each $v \in V$.
This dummy entry ensures that we find the same vertices, $n$, in the end when scanning two labels.
Thus we can avoid to separately test whether we have scanned till the end.

\myparagraph{Arrays}
For each label $L(v)$,
it is faster to store the array for vertices and the array for distances separately since distances are only used when vertices match~\cite{practice/road/Abraham2011}.
We also align arrays to cache lines.

\subsection{Theoretical Properties}
\label{sec:algorithm_theoretical}

\subsubsection{Minimality}

\vspace{-0.6em} \begin{theorem}
Let $L'_n$ be the index defined in Section~\ref{sec:algo_pruning}.
$L'_n$ is minimal in the sense that, for any vertex $v$ and for any pair $(u, \delta_{uv}) \in L'_n(v)$,
there is a pair of vertices $(s, t)$
such that, if we remove $(u, \delta_{uv})$ from $L'_n(v)$,
we cannot answer the correct distance between $s$ and $t$.
\end{theorem}

\begin{proof}
Let $v_i \in V$ and $(v_j, \delta_{v_j v_i}) \in L'_n(v_i)$.
This implies $j < i$.
We show that if we remove $(v_j, \delta_{v_j v_i})$ from $L'_n(v_i)$
then we cannot answer the correct distance between $v_i$ and $v_j$.
We claim that,
for any $k \neq j$,
either (i) $(v_k, \delta_{v_k v_i}) \not \in L'_n(v_i)$ or $(v_k, \delta_{v_k v_j}) \not\in L'_n(v_j)$ holds,
or (ii) $d_G(v_i, v_k) + d_G(v_k, v_j) > d_G(v_i, v_j)$ holds.
Suppose $k < j$ and assume that (ii) does not hold.
Then, (i) must hold since otherwise the $j$-th BFS should have pruned vertex $v_i$
  and $(v_j, \delta_{v_j v_i}) \not\in L'_n(v_i)$.
Suppose $k> j$ and assume that (ii) does not hold.
Then, $v_k \in P_G(v_i, v_j)$ and therefore $(v_j, \delta_{v_j v_k}) \in L'_j(v_k)$,
thus the $k$-th BFS prunes vertex $v_j$,
leading to $(v_k, \delta_{v_k v_j}) \not\in L'_n(v_j)$.
\end{proof}

\revision{
\subsubsection{Exploiting Existence of Highly Central Vertices}
\label{sec:algorithm_theoretical_landmark}
Then, we compare our method with landmark-based methods
to show that our method also can exploit the existence of highly central vertices.
We consider the standard landmark-based method~\cite{practice/landmarks/Potamias2009, practice/landmarks/Vieira2007},
which do not use any path heuristics.
As we stated in Section~\ref{sec:rel_approx},
by selecting central vertices as landmarks,
it attains remarkable average precision for real-world networks.
From the following theorem,
we can observe that
our method is efficient for networks whose distance can be answered by
landmark-based methods with such high precision,
and our method also can exploit the existence of these central vertices.

\begin{theorem}
\label{theorem:central}
If we assume that the standard landmark-based approximate method
can answer correct distances to $(1 - \epsilon)n^2$ pairs (out of $n^2$ pairs) using $k$ landmarks,
then the pruned landmark labeling method
gives an index with average label size $O(k + \epsilon n)$.
\end{theorem}

\begin{proof}[Sketch]
After conducting pruned BFSs from the $k$ landmark vertices first,
at most $\epsilon n^2$ pairs are added to the index in total,
since we never add pairs whose distance can be answered from current labels.
\end{proof}
}

\subsubsection{Exploiting Small Tree-width of Fringes}
\label{sec:algorithm:treewidth}
\revision{Finally, we show a theoretical evidence
that our method can also exploit tree-like fringes efficiently.}
As we mentioned in Section~\ref{sec:related_work:exact},
methods based on tree decompositions were proposed
for distance queries~\cite{practice/exact/Wei2010, practice/exact/Akiba2012}.
Both of them extend methods which work efficiently for graphs of small tree-width,
and they exploit low tree-width of fringes in real-world networks by tree decompositions.
Interestingly, though we do not use tree decompositions explicitly,
we can prove that our method can efficiently process graphs of small tree-width.
\revision{Thus, our method implicitly exploits this property of real-world networks too.}
For definitions of tree-width and tree decompositions, see~\cite{theory/td/Robertson1984}.

\begin{theorem}
\label{theorem:tree-width}
Let $w$ be the tree-width of $G$.
There is an order of vertices with which the pruned landmark labeling method
takes $O(wm \log n + w^2 n \log^2 n)$ time for preprocessing,
stores an index with $O(w n \log n)$ space,
and answers each query in $O(w \log n)$ time.
\end{theorem}

\begin{proof}[Sketch]
The key ingredient is the centroid decomposition~\cite{theory/centroid_decomposition} of the tree decomposition.
First we conduct pruned BFSs from all the vertices in a centroid bag.
Then, later pruned BFSs never go beyond the bag.
Therefore, we can consider as we divided the tree decomposition
into disjoint components, each having at most half of the bags.
We recursively repeat this procedure.
The number of recurrences is at most $O(\log n)$.
Since we add at most $w$ pairs to each vertex at each depth of recursion,
the number of pairs in each label is $O(w \log n)$.
At each depth of recursion,
the total time consumed by pruned BFSs from the current components is $O(wm + w^2 n \log n)$,
where $O(wm)$ is the time for traversing edges and $O(w^2 n \log n)$ is the time for pruning tests.
\end{proof}

\section{Bit-parallel Labeling}
\label{sec:bit-parallel}
To further speed up both preprocessing and querying,
we propose an optimizing method which exploits bit-level parallelism.
Bit-parallel methods are those that
perform different calculations on different bits in the same word
to exploit the fact that
computers can perform bitwise operations on a word at once. %
The word length is commonly 32 or 64 in computers of the day.

In the following, we denote the number of bits in a computer word as $b$
and assume bitwise operations on bit vectors of length $b$ can be done in $O(1)$ time.
We propose an algorithm to conduct BFSs and compute labels from $b+1$ roots
simultaneously in $O(m)$ time. %
Moreover, we also propose a method to answer
distance queries for any pair of vertices via one of these $b+1$ vertices in $O(1)$ time.

\subsection{Bit-parallel Labels}
To describe the preprocessing algorithm and the querying algorithm,
we first define what we store in the index.

As we explain in the next subsection,
we conduct bit-parallel BFSs from a vertex $r$ and a subset of its neighbors $S_r \subseteq N_G(r)$
with size at most $b$.
We define
\begin{align*}
S^{i}_r(v) = \set{u \in S_r \mid d_G(u, v) - d_G(r, v) = i}.
\end{align*}
Since vertices in $S_r$ are neighbors of $r$,
for any vertex $u \in S_r$ and any vertex $v \in V$,
$\abs{d_G(u, v) - d_G(r, v)} \leq 1$.
Therefore, for each $v \in V$,
$S_r$ can be partitioned into $S^{-1}_r(v), S^{0}_r(v)$, and $S^{+1}_r(v)$.
That is, $S^{-1}_r(v) \cup S^{0}_r(v) \cup S^{+1}_r(v) = S_r$.

\newcommand{\Lbp}{L_\textsc{BP}}

We compute \textit{bit-parallel labels} and store them in the index.
For each vertex $v \in V$, we precompute a bit-parallel label denoted as $\Lbp(v)$.
$\Lbp(v)$ is a set of quadruples $(u, \delta_{uv}, S^{-1}_u(v), S^{0}_u(v))$,
where $u \in V$ is a vertex, $\delta_{uv} = d_G(u, v)$ and $S^{i}_u(v) \subseteq S_u$ is defined above.
We store $S^{-1}_u(v)$ and $S^{0}_u(v)$ by bit vectors of $b$ bits.
Note that $S^{+1}_u(v)$ can be obtained as $S_u \setminus(S^{-1}_u(v) \cup S^{0}_u(v))$,
but actually we do not use $S^{+1}_u(v)$ in the querying algorithm.

In order to describe subsets of $S_r$ by bit vectors of $b$ bits,
we assign an unique number between one and $\abs{S_r}$ to each vertex in $S_r$,
and express presence of the $i$-th vertex by setting the $i$-th bit.

\subsection{Bit-parallel BFS}
\begin{algorithm}[t]
\caption{Bit-parallel BFS from $r \in V$ and $S_r \subseteq N_G(r)$.}
\label{alg:bp}
\begin{algorithmic}[1]
\Procedure{Bp-BFS}{$G$, $r$, $S_r$}
   \State $(P[v], S^{-1}_r[v], S^{0}_r[v]) \gets (\infty, \emptyset, \emptyset)$
    for all $v \in V$
   \State $(P[r], S^{-1}_r[r], S^{0}_r[r]) \gets (0, \emptyset, \emptyset)$
   \State $(P[v], S^{-1}_r[v], S^{0}_r[v]) \gets (1, \set{v}, \emptyset)$ for all $v \in S_r$
  \State $Q_0, Q_1 \gets $ an empty queue
  \State Enqueue $r$ onto $Q_0$
  \State Enqueue $v$ onto $Q_1$ for all $v \in S_r$
  \While{$Q_0$ is not empty}
    \State $E_0 \gets \emptyset$ and $E_1 \gets \emptyset$
    \While{$Q_0$ is not empty}
      \State Dequeue $v$ from $Q_0$.
      \ForAll{$u \in N_G(v)$}
        \If{$P[u] = \infty \lor P[u] = P[v] + 1$}
          \State $E_1 \gets E_1 \cup \set{(v, u)}$
          \If{$P[u] = \infty$}
            \State $P[u] \gets P[v] + 1$
            \State Enqueue $u$ onto $Q_1$.
          \EndIf
        \ElsIf{$P[u] = P[v]$}
          \State $E_0 \gets E_0 \cup \set{(v, u)}$
        \EndIf
      \EndFor
    \EndWhile
    \ForAll{$(v, u) \in E_0$}
      \State $S^{0}_r[u] \gets S^{0}_r[u] \cup S^{-1}_r[v]$
    \EndFor
    \ForAll{$(v, u) \in E_1$}
      \State $S^{-1}_r[u] \gets S^{-1}_r[u] \cup S^{-1}_r[v]$
      \State $S^{0}_r[u] \gets S^{0}_r[u] \cup S^{0}_r[v]$
    \EndFor
    \State $Q_0 \gets Q_1$ and $Q_1 \gets \emptyset$
  \EndWhile
  \State \Return $(P, S^{-1}_r, S^{0}_r)$
\EndProcedure
\end{algorithmic}
\end{algorithm}

We once put aside the pruning discussed in Section~\ref{sec:algo_pruning}
and we make a bit-parallel version of the naive labeling method
discussed in Section~\ref{sec:algo_naive}.
We introduce pruning later in Section~\ref{sec:bp_pruning}.

Let $r \in V$ be a vertex and $S_r \subseteq N_G(r)$ be a subset of neighbors of $r$ with size at most $b$.
We explain an algorithm to compute $d_G(r, v)$, $S^{-1}_r(v)$ and $S^{0}_r(v)$ for all $v \in V$ that are reachable from $\set{r} \cup S_r$.
The algorithm is described as Algorithm~\ref{alg:bp}.
Basically we conduct a BFS from $r$
computing sets $S^{-1}$ and $S^{0}$.

Let $v$ be a vertex.
Suppose that we have already computed
$S^{-1}_r(w)$ for all $w$ such that $d_G(r, w) < d_G(r, v)$.
We can compute $S^{-1}_r(v)$ as follows,
\begin{align*}
\set{u \in S_r \mid u \in S^{-1}_r(w), w \in N_G(v), d_G(r, w) = d_G(r, v) - 1 },
\end{align*}
since if $u$ is in $S^{-1}_r(v)$,
$d_G(u, v) = d_G(r, v) - 1$ and therefore $u$ is on one of the shortest paths
from $r$ to $v$.
Similarly, assuming that we have already computed
$S^{-1}_r(w)$ for all $w$ such that $d_G(r, w) \leq d_G(r, v)$
and $S^{0}_r(w)$ for all $w$ such that $d_G(r, w) < d_G(r, v)$,
we can compute $S^{0}_r(v)$ as follows,
\begin{align*}
\set{u \in S_r \mid u \in S_r^{0}(w), w \in N_G(v), d_G(r, w) = d_G(r, v) - 1 } \\
\cup \set{u \in S_r \mid u \in S_r^{-1}(w), w \in N_G(v), d_G(r, w) = d_G(r, v)}.
\end{align*}

Therefore, along with the breadth-first search,
we can compute $S^{-1}_r$ and $S^{0}_r$ alternately by dynamic programming
in the increasing order of distance from $r$.
That is,
first we compute $S^{-1}_r(v)$ for all $v \in V$ such that $d_G(r, v) = 1$,
next we compute $S^{0}_r(v)$ for all $v \in V$ such that $d_G(r, v) = 1$,
then we compute $S^{-1}_r(v)$ for all $v \in V$ such that $d_G(r, v) = 2$,
next we compute $S^{0}_r(v)$ for all $v \in V$ such that $d_G(r, v) = 2$,
and so on.
Note that operations on sets can be done in $O(1)$ time by
representing sets by bit vectors and using bitwise operations.

\subsection{Bit-parallel Distance Querying}
\label{sec:bit-parallel_query}
\revision{
To process a distance query between a pair of vertices $s$ and $t$,
as with normal labels, we scan bit-parallel labels $\Lbp(s)$ and $\Lbp(t)$.
For each pair of quadruples that share the same root vertex,
$(r, \delta_{rs}, S^{-1}_r(s), S^{0}_r(s)) \in \Lbp(s)$ and
$(r, \delta_{rt}, S^{-1}_r(t), S^{0}_r(t)) \in \Lbp(t)$,
from these quadruples we compute distance between $s$ and $t$
via one of vertices in $\set{r} \cup S_r$.
That is, we compute $\delta = \min\limits_{u \in \set{r}\cup S_r} \set{d_G(s, u) + d_G(u, t)}$.
}%
A naive way is to compute $d_G(s, u)$ and $d_G(u, t)$
for all $u$ and take the minimum,
which takes $O(\abs{S_r})$ time.
However, we propose an algorithm to compute $\delta$
in $O(1)$ time by exploiting bitwise operations.

Let $\tilde{\delta} = d_G(s, r) + d_G(r, t)$.
Since $\tilde{\delta}$ is an upper bound on $\delta$
and $d_G(s, u) \geq d_G(s, r) - 1, d_G(u, t) \geq d_G(r, t) - 1$ for all $u \in S_r$,
$\tilde{\delta} - 2 \leq \delta \leq \tilde{\delta}$.
Therefore, what we have to do is to judge
whether the distance $\delta$ is $\tilde{\delta}-2$, $\tilde{\delta}-1$ or $\tilde{\delta}$.

This can be done as follows.
If $S^{-1}_r(s) \cap S^{-1}_r(t) \not= \emptyset$,
then $\delta = \tilde{\delta}-2$.
Otherwise, if $S^{0}_r(s) \cap S^{-1}_r(t) \not= \emptyset$
or $S^{-1}_r(s) \cap S^{0}_r(t) \not= \emptyset$,
then $\delta = \tilde{\delta}-1$, and otherwise $\delta = \tilde{\delta}$.
Note that computing intersections of sets can be done by bitwise AND operations.
Therefore, all these operations can be done in $O(1)$ time.
\revision{
Thus, the distance $\delta$ can be computed in $O(1)$ time,
and, in total, we can answer each query in $O(\abs{\Lbp(s)} + \abs{\Lbp(t)})$ time.
}

\subsection{Introducing to Pruned Labeling}
\label{sec:bp_pruning}
Now we discuss how to combine this bit-parallel labeling methods
and the pruned labeling method discussed in Section~\ref{sec:algo_pruning}.
We propose a simple and efficient way as follows.
First we conduct bit-parallel BFSs without pruning for $t$ times, where $t$ is a parameter.
Then, we conduct pruned BFSs using both the bit-parallel labels and normal labels for pruning.

This method exploits different strength of the pruned labeling method and the bit-parallel labeling method.
In the beginning, pruning does not work much and pruned BFSs visits large portion of the vertices.
Therefore, instead of pruned labeling, we use bit-parallel labeling without pruning to efficiently cover a larger part of pairs of vertices.
Skipping the overhead of vain pruning tests also contributes the speed-up.

As roots and neighbor sets for bit-parallel BFSs,
we propose to greedily use vertices with the highest priority:
we select a vertex with the highest priority as the root $r$ among remaining vertices,
and we select up to $b$ vertices with the highest priority as the set $S_r$ among remaining neighbors.%

As we see in the experimental results in Section~\ref{sec:experiments},
this method improves the preprocessing time, the index size and the query time.
Moreover, as we also confirm in the experiments,
if we do not set too large value as $t$,
at least it does not spoil the performance.
Therefore we do not have to be too serious about finding a proper value for $t$,
and our method is still easy to use.

\section{Variants and Extensions}
\label{sec:variants}
\noindent \textbf{Shortest-Path Queries:}
To answer not only distances but also shortest-paths,
we store sets of tuples instead of pairs as labels.
Label $L(v)$ is a set of triples $(u, \delta_{uv}, p_{uv})$,
where $p_{uv} \in V$ is the parent of $u$ in the pruned breadth-first search tree rooted at $u$
created by the pruned BFS from $u$.
We can restore the shortest path between $v$ and $u$ by ascending the tree from $v$ to the parents.

\myparagraph{Weighted Graphs}
To treat weighted graphs,
the only necessary change is to perform pruned Dijkstra's algorithm instead of pruned BFSs.
Bit-parallel labeling cannot be used for weighted graphs.

\newcommand{\Lin}{L_\textsc{IN}}
\newcommand{\Lout}{L_\textsc{OUT}}

\myparagraph{Directed Graphs}
To treat directed graphs,
we first redefine $d_G(u,v)$ as the distance from $u$ to $v$.
Then, we store two labels $\Lout(v)$ and $\Lin(v)$ for each vertex.
Label $\Lout(v)$ is a set of pairs $(u, \delta_{vu})$,
where $u \in V$ and $\delta_{vu} = d_G(v, u)$,
and Label $\Lin(v)$ is a set of pairs $(u, \delta_{uv})$,
where $u \in V$ and $\delta_{uv} = d_G(u, v)$.
We can answer the distance from vertex $s$ to vertex $t$ by $\Lout(s)$ and $\Lin(t)$.
To compute these labels,
from each vertex, we conduct pruned BFSs twice:
once in the forward direction and once in the reverse direction.

\revision{
\myparagraph{Disk-based Query Answering}
To answer a distance query,
our querying algorithm only refers to two contiguous regions.
Thus, if the index is disk resident,
we can answer queries with two disk seek operations,
which would be still much faster than an in-memory BFS.
}

\section{Experiments}
\label{sec:experiments}

\begin{figure*}[t]
  \centering
  \subfloat[Degree distribution of smaller five datasets.]{
    \hspace{-1.5em}
    \epsfig{file=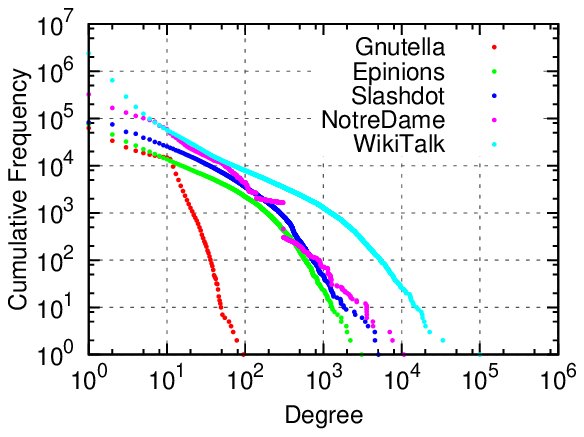, height=32mm}
    \label{fig:degree1}
    \hspace{-1.5em}
  }
  \hspace{0.5em}
  \subfloat[Degree distribution of larger six datasets.]{
    \hspace{-1.5em}
    \epsfig{file=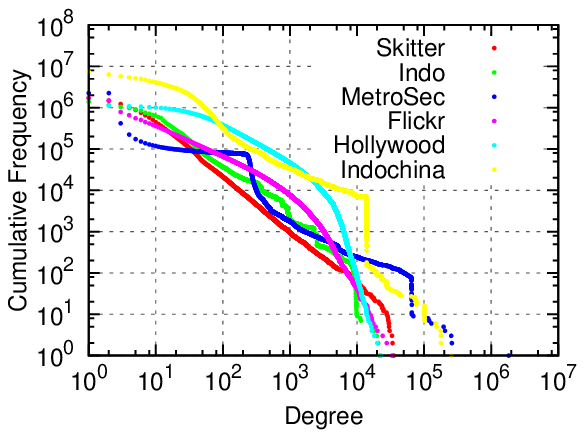, height=32mm}
    \label{fig:degree2}
    \hspace{-1.5em}
  }
  \hspace{0.5em}
  \subfloat[Distance distribution of smaller five datasets.]{
    \hspace{-1.5em}
    \epsfig{file=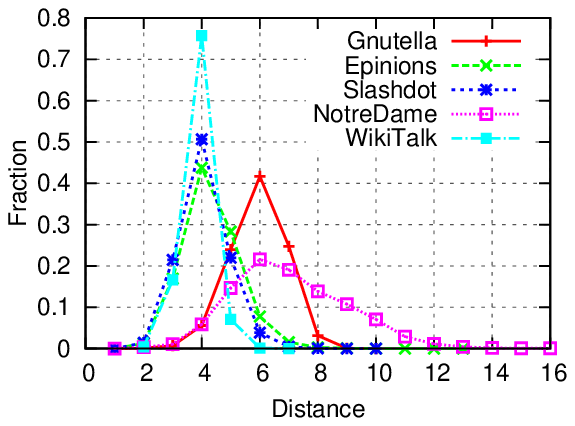, height=32mm}
    \label{fig:distance1}
    \hspace{-1.5em}
  }
  \hspace{0.5em}
  \subfloat[Distance distribution of larger six datasets.]{
    \hspace{-1.5em}
    \epsfig{file=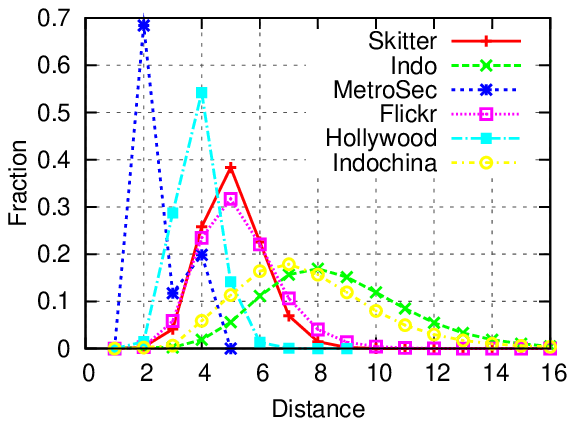, height=32mm}
    \label{fig:distance2}
    \hspace{-1.5em}
  }
  \caption{Properties of the datasets.}
\end{figure*}

{
\tabcolsep = 1.7mm
\begin{table*}[htbp]
\small\center
\caption{Performance comparison between the proposed method and previous methods for the real-world datasets.
\textsc{IT} denotes indexing time, \textsc{IS} denotes index size,
\textsc{QT} denotes query time, and \textsc{LN} denotes average label size for each vertex.
\textsc{DNF} means it did not finish in one day or ran out of memory.}
\begin{tabular}{|l||rrrr|rrrr|rrr|r|}
\hline
 \multirow{2}*{\hfill Dataset \hfill}
 & \multicolumn{4}{c|}{Pruned Landmark Labeling}
 & \multicolumn{4}{c|}{Hierarchical Hub Labeling~\cite{practice/exact/Abraham2012}}
 & \multicolumn{3}{c|}{Tree Decomposition~\cite{practice/exact/Akiba2012}}
 & \multirow{2}*{\hfill BFS \hfill} \\
 & \textsc{IT} & \textsc{IS} & \textsc{QT} & \textsc{LN} &
\textsc{IT} & \textsc{IS} & \textsc{QT} & \textsc{LN} &
\textsc{IT}  & \textsc{IS} & \textsc{QT} & \\
\hline\hline
Gnutella & 54 s & 209 MB & 5.2 $\mu$s & 437+16 & 245 s & 380 MB & 11 $\mu$s & 1,275 & 209 s & 68 MB & 19 $\mu$s & 3.2 ms \\
Epinions & 1.7 s & 32 MB & 0.5 $\mu$s & 7+16 & 495 s & 93 MB & 2.2 $\mu$s & 256 & 128 s & 42 MB & 11 $\mu$s & 7.4 ms \\
Slashdot & 6.0 s & 48 MB & 0.8 $\mu$s & 14+16 & 670 s & 182 MB & 3.9 $\mu$s & 464 & 343 s & 83 MB & 12 $\mu$s & 12 ms \\
NotreDame & 4.5 s & 138 MB & 0.5 $\mu$s & 29+16 & 10,256 s & 64 MB & 0.4 $\mu$s & 41 & 243 s & 120 MB & 39 $\mu$s & 17 ms \\
WikiTalk & 61 s & 1.0 GB & 0.6 $\mu$s & 9+16 & \textsc{DNF} & - & - & - & 2,459 s & 416 MB & 1.8 $\mu$s & 197 ms \\
Skitter & 359 s & 2.7 GB & 2.3 $\mu$s & 123+64 & \textsc{DNF} & - & - & - & \textsc{DNF} & - & - & 190 ms \\
Indo & 173 s & 2.3 GB & 1.6 $\mu$s & 133+64 & \textsc{DNF} & - & - & - & \textsc{DNF} & - & - & 150 ms \\
MetroSec & 108 s & 2.5 GB & 0.7 $\mu$s & 19+64 & \textsc{DNF} & - & - & - & \textsc{DNF} & - & - & 150 ms \\
Flickr & 866 s & 4.0 GB & 2.6 $\mu$s & 247+64 & \textsc{DNF} & - & - & - & \textsc{DNF} & - & - & 361 ms \\
Hollywood & 15,164 s & 12 GB & 15.6 $\mu$s & 2,098+64 & \textsc{DNF} & - & - & - & \textsc{DNF} & - & - & 1.2 s \\
Indochina & 6,068 s & 22 GB & 4.1 $\mu$s & 415+64 & \textsc{DNF} & - & - & - & \textsc{DNF} & - & - & 1.5 s \\
\hline
\end{tabular}
\label{tbl:performance}
\vspace{-1em}
\end{table*}
}

\begin{figure*}[t]
  \centering
  \subfloat[Number of vertices labeled in each pruned BFS.]{
    \hspace{-1.5em}
    \epsfig{file=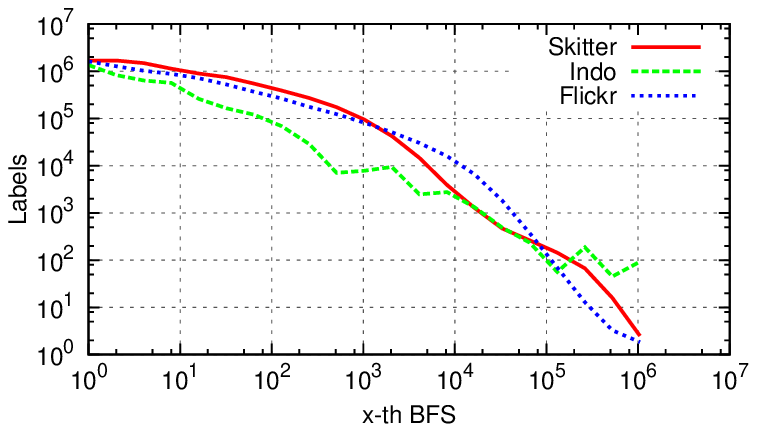, height=32mm}
    \label{fig:pruned_bfs}
    \hspace{-1.5em}
  }
  \hspace{1.5em}
  \subfloat[Cumulative distribution of the number of vertices labeled in each pruned BFS.]{
    \hspace{-1.5em}
    \epsfig{file=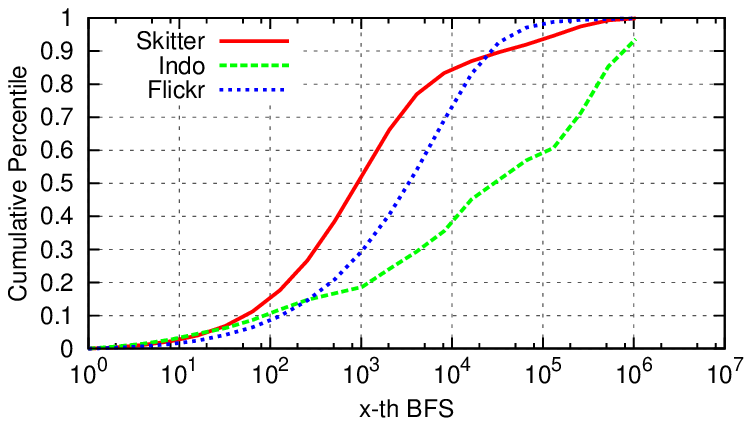, height=32mm}
    \label{fig:pruned_bfs_cumulative}
    \hspace{-1.5em}
  }
  \hspace{1.5em}
  \subfloat[Distribution of the sizes of labels.]{
    \hspace{-1.5em}
    \epsfig{file=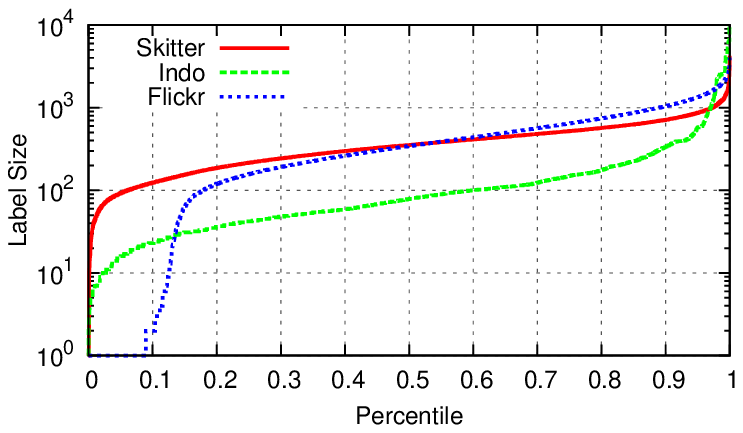, height=32mm}
    \label{fig:labels}
    \hspace{-1.5em}
  }
  \caption{Effect of pruning and sizes of labels.}
\end{figure*}

\begin{figure*}[t]
  \centering
  \subfloat[Average]{
    \hspace{-1.5em}
    \epsfig{file=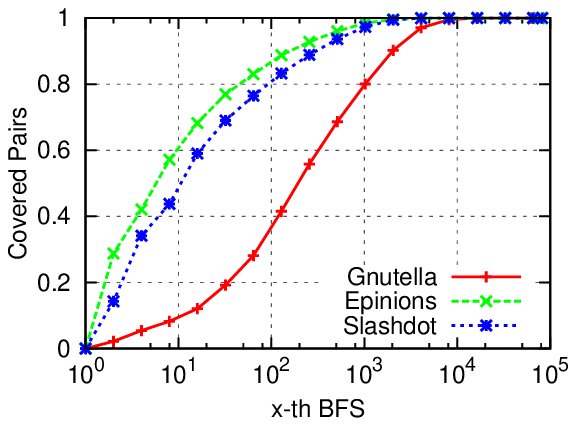, height=32mm}
    \label{fig:coverage_avg}
    \hspace{-1.5em}
  }
  \hspace{0.5em}
  \subfloat[Distance-wise, Gnutella]{
    \hspace{-1.5em}
    \epsfig{file=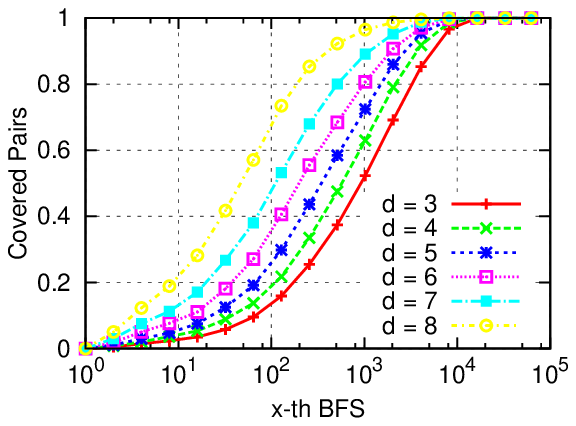, height=32mm}
    \label{fig:coverage_gnutella}
    \hspace{-1.5em}
  }
  \hspace{0.5em}
  \subfloat[Distance-wise, Epinions]{
    \hspace{-1.5em}
    \epsfig{file=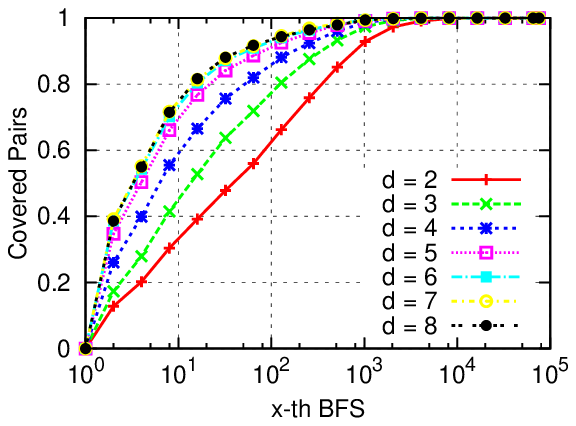, height=32mm}
    \label{fig:coverage_epinions}
    \hspace{-1.5em}
  }
  \hspace{0.5em}
  \subfloat[Distance-wise, Slashdot]{
    \hspace{-1.5em}
    \epsfig{file=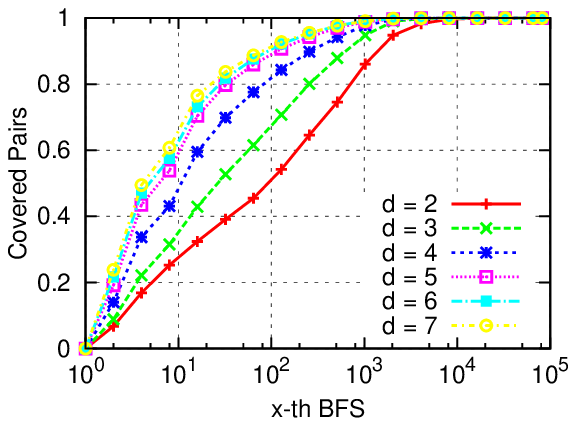, height=32mm}
    \label{fig:coverage_slashdot}
    \hspace{-1.5em}
  }
  \caption{Fraction of pairs of vertices whose distance can be answered by index, against
number of performed pruned BFS.}
\vspace{-1em}
\end{figure*}

\begin{figure*}[t]
  \centering
  \subfloat[Preprocessing time]{
    \hspace{-1.5em}
    \epsfig{file=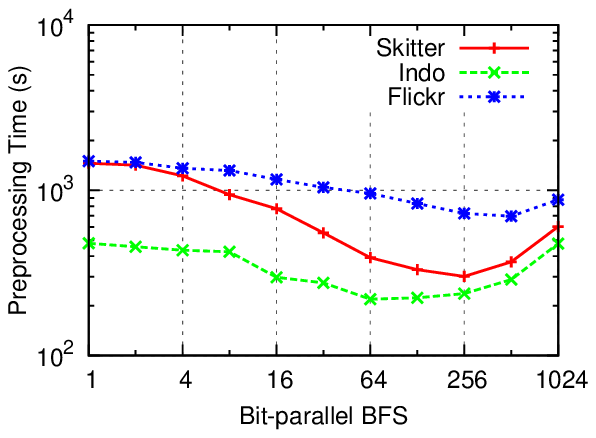, height=32mm}
    \label{fig:bp_time}
    \hspace{-1.5em}
  }
  \hspace{0.5em}
  \hspace{1.5em}
  \subfloat[Query time]{
    \hspace{-3em}
    \epsfig{file=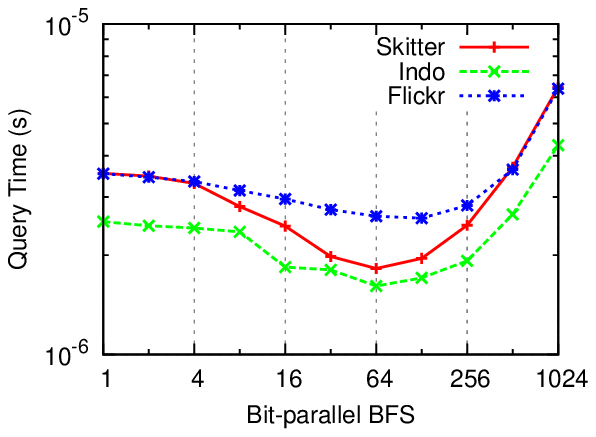, height=32mm}
    \label{fig:bp_query}
    \hspace{-3em}
  }
  \hspace{0.5em}
  \subfloat[Average size of a normal label]{
    \hspace{0em}
    \epsfig{file=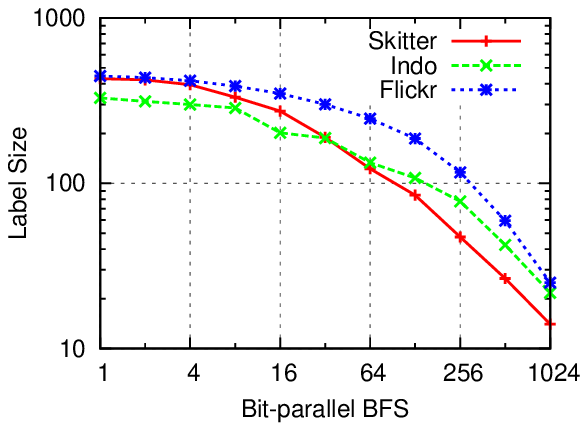, height=32mm}
    \label{fig:bp_labels}
    \hspace{1em}
  }
  \hspace{0.5em}
  \subfloat[Index size]{
    \hspace{-4em}
    \epsfig{file=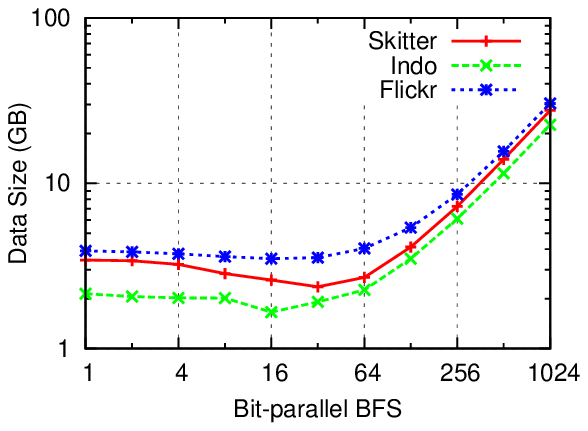, height=32mm}
    \label{fig:bp_size}
    \hspace{-4em}
  }
  \hspace{2.5em}
  \caption{Performance against number of bit-parallel BFSs.}
  \label{fig:bp}
\vspace{-1em}
\end{figure*}

We conducted experiments on a Linux server with
Intel Xeon X5670 (2.93 GHz) and 48GB of main memory.
The proposed method was implemented in C++.
We used 8-bit integers to represent distances,
32-bit integers to represent vertices,
and 64-bit integers to conduct bit-parallel BFSs.
For vertex ordering, we mainly use the \textsc{Degree} strategy
and we do not specify the vertex ordering strategy
unless we use other strategies.
For query time, we generally report the average time for 1,000,000 random queries.

\subsection{Datasets} %
\begin{table}[tbp]
\small\center
\caption{Datasets}
\begin{tabular}{|l|lrr|}
\hline
Dataset & Network & $\abs{V}$ & $\abs{E}$ \\
\hline\hline
Gnutella & Computer & 63 K & 148 K \\
Epinions & Social & 76 K & 509 K \\
Slashdot & Social & 82 K & 948 K \\
Notredame & Web & 326 K & 1.5 M \\
WikiTalk & Social & 2.4 M & 4.7 M \\
Skitter & Computer & 1.7 M & 11 M \\
Indo & Web & 1.4 M & 17 M \\
MetroSec & Computer & 2.3 M & 22 M \\
Flickr & Social & 1.8 M & 23 M \\
Hollywood & Social & 1.1 M & 114 M \\
Indochina & Web & 7.4 M & 194 M \\
\hline
\end{tabular}
\label{tbl:dataset}
\vspace{-1em}
\end{table}

To show the efficiency and robustness of our method,
we conducted experiments on various real-world networks:
five social networks, three web graphs and three computer networks.
We treated all the graphs as undirected, unweighted graphs.
Basically we used five smaller datasets to compare the performance between the proposed
method and previous methods and to analyze the behavior of these methods,
and used larger six datasets to show the scalability of the proposed method.
The types of networks, the numbers of vertices and edges
are presented in Table~\ref{tbl:dataset}.

\subsubsection{Detailed Description}

\iftrue%
\newcommand{\dataset}[1]{\vspace{0.3em}\noindent {\textbf{#1:}}}

\dataset{Gnutella}
This is a graph created from a snapshot of the Gnutella P2P network in August 2002~\cite{dataset/gnutella}.

\dataset{Epinions}
This graph is the on-line social network in Epinions (\texttt{www.epinions.com}),
where each vertex represents a user and
each edge represents a trust relationship~\cite{dataset/epinions}.

\dataset{Slashdot}
This is the on-line social network in Slashdot (\texttt{slashdot.org}) obtained in February 2009.
Vertices correspond to users and edges correspond to friend/foe links between the users~\cite{dataset/slashdot}.

\dataset{NotreDame}
This is a web graph between pages from University of Notre Dame (domain \texttt{nd.edu})
collected in 1999~\cite{dataset/notredame}.

\dataset{WikiTalk}
This is the on-line social network among editors of Wikipedia (\texttt{www.wikipedia.org})
created by communication on edits on talk pages by till January 2008~\cite{dataset/wikitalk1, dataset/wikitalk2}.

\dataset{Skitter}
This is an Internet topology graph created from
traceroutes run in 2005 by Skitter~\cite{dataset/skitter}.

\dataset{Indo}
This is a web graph between pages in \texttt{.in} domain
crawled in 2004~\cite{dataset/webgraph1, dataset/webgraph2}.

\dataset{MetroSec}
This is a graph constructed from Internet traffic captured by MetroSec.
Each vertex represents a computer and two vertices are linked
if they appear in a packet as sender and destination~\cite{dataset/metrosec}.

\dataset{Flickr}
This is the on-line social network in a photo-sharing site,
Flickr (\texttt{www.flickr.com})\cite{dataset/mpi}.

\dataset{Indochina}
This is a web graph of web pages in the country domains of Indochina countries,
crawled in 2004~\cite{dataset/webgraph1, dataset/webgraph2}.

\dataset{Hollywood}
This is a social network of movie actors.
Two actors are linked if they appeared in a movie together
by 2009~\cite{dataset/webgraph1, dataset/webgraph2}.
\fi%

\subsubsection{Statistics}

First,
we %
investigated the degree distribution of the networks,
since degrees of vertices play important roles in our method when we use \textsc{Degree} strategy for vertex ordering.
Figures~\ref{fig:degree1} and~\ref{fig:degree2}
are the log-log plot of degree complementary cumulative distribution.
As expected, we can confirm that all these networks generally exhibit power-law degree distributions.

Then,
we also examined the distribution of distances.
Figures~\ref{fig:distance1} and~\ref{fig:distance2}
show distribution of distances for 1,000,000 random pairs of vertices.
As we can observe from these figures,
these networks are also small-world networks,
in the sense that the average distance is very small.

\subsection{Performance}
\label{sec:experiments_performance}

First we present the performance of our method on the real-world datasets
to show the efficiency and robustness of our method.
Table~\ref{tbl:performance} shows the performance of our method for the datasets.
\textsc{IT} denotes preprocessing time, \textsc{IS} denotes index size,
\textsc{QT} denotes average query time for 1,000,000 random queries,
and \textsc{LN} denotes the average label size for each vertex,
in the format of the size of normal labels (left) plus the size of bit-parallel labels (right).
We set the number of times we conduct bit-parallel BFSs
as 16 for first five datasets and 64 for the rest.

In Table~\ref{tbl:performance},
we also listed the performance of two of the state-of-the-art existing methods.
One is \textit{hierarchical hub labeling}~\cite{practice/exact/Abraham2012},
which is also based on distance labeling.
The other one is based on tree decompositions~\cite{practice/exact/Akiba2012},
which is an improved version of TEDI~\cite{practice/exact/Wei2010}.
For these previous methods, we used the implementations by the authors of these methods, both in C++.
Experiments for hierarchical hub labeling
were conducted on a Windows server with two Intel Xeon X5680 (3.33GHz) and 96GB of main memory.
Experiments for the tree-decomposition-based method were conducted on our environment described above.
We also described the average time to compute distance by breadth-first search
for 1,000 random pairs of vertices.
Among these four methods including the proposed method,
only the preprocessing of hierarchical hub labeling~\cite{practice/exact/Abraham2012} was parallelized to use all the 12 cores.
All the other timing results are sequential.

\subsubsection{Preprocessing Time and Scalability}
Our emphasis is particularly on the large improvement in the preprocessing time,
leading to much better scalability.
First, we successfully preprocessed the largest two datasets Hollywood and Indochina
with millions of vertices and hundreds of millions of edges in moderate preprocessing time%
.
This is improvement of two orders of magnitude on the graph size we can handle
since,
as we listed in Table~\ref{tbl:methods},
other existing exact distance querying methods
take thousands or tens of thousands of seconds to preprocess graphs with millions of edges.

For next four datasets with tens of millions of edges,
it took less than one thousand seconds,
while the previous methods did not finish after one day or ran out of memory.
For smaller six datasets, they took at most one minute,
and about at least 50 times faster than the previous methods for the most of them.

\subsubsection{Query Time}
The average query time was generally microseconds and at most 16 microseconds.
For almost all the smaller five datasets,
the query time of the proposed method is faster than the query time of the previous methods.
Indeed, from Table~\ref{tbl:methods}, we can also observe that
the query time of our method is comparable to all the existing methods for graphs of these sizes.
Moreover, we can confirm that the query time does not increase much for larger networks.

\subsubsection{Index Size}
As for the smaller five networks,
results demonstrate that our method is comparable to the previous methods with respect to index size.
However,
even though nowadays computers with tens of gigabytes of memory are neither rare nor expensive,
reducing the index size can be an important next research issue.

\subsection{Analysis}%
Next we analyze the behavior of our method
to investigate why our method is efficient.

\subsubsection{Pruned BFS}
\label{sec:experiments_pruned_bfs}

First we study how labels are computed and stored.
Figure~\ref{fig:pruned_bfs} shows the number of distances
added to labels in each pruned BFS, and Figure~\ref{fig:pruned_bfs_cumulative}
shows the cumulative distribution of it, that is,
the ratio of the distances stored no later than each step
to all the distances stored in the end.
We did not use bit-parallel BFSs for these experiments.

From these figures, we can confirm the large impact of the pruning.
Figure~\ref{fig:pruned_bfs} shows that the number of distances added to labels in each BFS decreases so rapidly.
For example, after 1,000 times of BFSs,
for all the three datasets distances are added to the labels of only less than 10\% of the vertices,
and after conducting 10,000 times of BFSs,
for all the three datasets distances are added to the labels of only less than 1\% of the vertices.
Figure~\ref{fig:pruned_bfs_cumulative} also shows that
large portion of the labels are computed in the beginning.

\subsubsection{Sizes of Labels}
Figure~\ref{fig:labels} shows the distribution of the sizes of labels after the whole preprocessing,
sorted in the ascending order of sizes.
We can observe that the size of a label each vertex has
do not differ much for different vertices,
and few vertices have much larger labels than the average.
This shows that the query time of our method is quite stable.

If you are anxious about vertices with unusually large labels,
you can precompute the distance between these vertices and all the vertices
and answer it directly,
since the number of such vertices are few as shown in Figure~\ref{fig:labels}.

\subsubsection{Pair Coverage}
Figure~\ref{fig:coverage_avg}
illustrates the ratio of the \textit{covered} pairs of vertices,
that is, the pairs of vertices
whose distances can be answered correctly by current labels,
at each step.
We used 1,000,000 random pairs to estimate these ratios.
We can observe that most pairs are covered in the beginning.
This shows that such a large portion of pairs
have the shortest paths that pass
such a small portion of central vertices, which are selected by the \textsc{Degree} strategy.
This is the reason why landmark-based approximate methods have good precision,
and also the reason why our pruning works so effectively.

Figures~\ref{fig:coverage_gnutella}, \ref{fig:coverage_epinions} and \ref{fig:coverage_slashdot}
illustrate the ratio of the covered pairs of vertices at each step
with pairs classified by distance.
They show that
generally distant pairs are covered earlier than close pairs.
This is the reason why the precision of landmark-based approximate methods
for close pairs are far worse than the precision for distant pairs.
On the other hand, our method aggressively exploits this property:
because distant pairs are covered in the beginning,
we can prune distant vertices when processing other vertices,
which results in fast preprocessing.

\subsubsection{Vertex Ordering Strategies}
\label{sec:exp_ordering}

\begin{table}[tbp]
\center\small
\caption{Average size of a label for each vertex against different vertex ordering strategies.}
\begin{tabular}{|l|rrr|}
\hline
Dataset & \textsc{Random} & \textsc{Degree} & \textsc{Closeness} \\
\hline\hline
Gnutella & 6,171 & 781 & 865 \\
Epinions & 7,038 & 124 & 132 \\
Slashdot & 8,665 & 216 & 234 \\
NotreDame & \textsc{DNF} & 60 & 82 \\
WikiTalk & \textsc{DNF} & 118 & 158 \\
\hline
\end{tabular}
\label{tbl:strategies}
\vspace{-1em}
\end{table}

Next we see the effect of vertex ordering strategies.
Table~\ref{tbl:strategies} describes
the average size of a label for each vertex
using different vertex ordering strategies described in Section~\ref{sec:vertex_order}.
We did not use bit-parallel BFSs for these experiments.
As we can see, results are not so different between the \textsc{Degree} strategy
and the \textsc{Closeness} strategy. The \textsc{Degree} strategy might be slightly better.
On the other hand, the result of the \textsc{Random} strategy is much worse than other two strategies.
This shows that by the \textsc{Degree} and \textsc{Closeness} strategies
we can successfully capture central vertices.

\subsubsection{Bit-parallel BFS}
Finally, we see the effect of bit-parallel BFSs
discussed in Section~\ref{sec:bit-parallel}.
Figure~\ref{fig:bp} shows the performance of our method
against different number of times we conduct bit-parallel BFSs.

Figure~\ref{fig:bp_time} illustrates preprocessing time.
It shows that, with a proper number of bit-parallel BFSs,
preprocessing time gets two to ten times faster,
resulting in the further enhancement to the scalability of our method.
Figure~\ref{fig:bp_query} illustrates query time.
We can confirm that query time also gets faster.
Figure~\ref{fig:bp_labels} shows the average size of a normal label for each vertex.
As we increase the number of bit-parallel BFSs,
many pairs are covered by special labels computed by bit-parallel BFSs,
and the size of normal labels decreases.
Figure~\ref{fig:bp_size} shows the index size.
With a proper number of bit-parallel BFSs, index size also decreases. %

Another important finding from these figures is
that the performance of our method is not too sensitive to
the parameter of the number of bit-parallel BFSs.
As they show,
the performance of our method does not become worse much unless we choose a too big number.
The proper parameters seem to common between different networks.
Therefore,
our method still is easy to use with this bit-parallel technique.

\section{Conclusions}
\label{sec:conclusions}

In this paper, we proposed a novel and efficient method
for exact shortest-path distance queries on large graphs.
Our method is based on distance labeling to vertices,
which is common to the existing exact distance querying methods,
but our labeling algorithm stands on a totally new idea.
Our algorithm conducts breadth-first search (BFS) from all the vertices with pruning.
Though the algorithm is simple, our pruning surprisingly reduce the search space and the labels,
resulting in fast preprocessing time, small index size and fast query time.
Moreover, we also proposed another labeling scheme exploiting bit-level parallelism,
which can be easily combined with the pruned labeling method to further improve the performance.
Extensive experimental results on large-scale real-world networks of various types
demonstrated the efficiency and robustness of our methods.
In particular, our method can handle networks with hundreds of millions of vertices,
which are two orders of magnitude larger than the limits of the previous methods,
with comparable index size and query time.

\revision{
We plan to investigate ways to handle even larger graphs,
where indices and/or graphs might not fit in main memory.
The first way is to reduce the index size
by reducing graphs exploiting obvious parts and symmetry~\cite{practice/landmarks/Qiao2012, rel/graph_compression}
and compressing labels by making dictionaries of common subtrees for shortest path trees~\cite{practice/road/Abraham2011}.
Another way is disk-based or distributed implementation.
As we stated in Section~\ref{sec:variants},
disk-based query answering is obvious and ready,
and the challenges are particularly on preprocessing.
However, since our preprocessing algorithm is a simple algorithm based on BFS,
we can leverage the large body of existing work on BFS.
In particular,
since pruning can be done locally,
the preprocessing algorithm would perform well on
BSP-model-based distributed graph processing platforms~\cite{rel/pregel}.

}

\section{Acknowledgments}
We would like to thank Hiroshi Imai for critical reading of the manuscript,
and Daniel Delling for providing us the experimental results of hierarchical hub labeling~\cite{practice/exact/Abraham2012}.
We would also like to thank the anonymous reviewers for their constructive
suggestions on improving the paper.
Yuichi Yoshida is supported by
JSPS Grant-in-Aid for Research Activity Start-up (24800082), MEXT Grant-in-Aid for Scientific Research on Innovative Areas (24106001), and JST, ERATO, Kawarabayashi Large Graph Project.

\bibliographystyle{abbrv}
\bibliography{204}

\end{document}